\documentclass[pdflatex,sn-mathphys-num]{sn-jnl}

\usepackage{amsmath,amssymb,amsthm,mathtools}
\usepackage{xurl}
\usepackage{microtype}

\theoremstyle{plain}
\newtheorem{theorem}{Theorem}[section]
\newtheorem{lemma}[theorem]{Lemma}
\newtheorem{proposition}[theorem]{Proposition}
\newtheorem{corollary}[theorem]{Corollary}
\theoremstyle{definition}

\newtheorem{assumption}[theorem]{Assumption}
\newtheorem{example}[theorem]{Example}
\theoremstyle{remark}
\newtheorem{remark}[theorem]{Remark}

\newcommand{\TV}{\operatorname{TV}}
\newcommand{\Prob}{\mathbb{P}}
\newcommand{\Expect}{\mathbb{E}}
\newcommand{\ind}{\mathbf{1}}
\newcommand{\Ddev}{\mathcal D_{\mathrm{dev}}}
\makeatletter
\g@addto@macro\ps@headings{\let\@oddhead\@empty\let\@evenhead\@empty}
\g@addto@macro\ps@titlepage{\let\@oddhead\@empty\let\@evenhead\@empty}
\makeatother

\begin{document}

\title[Conformal Within-Document Screening]{Two Conformal Constructions for Adaptive Within-Document AI-Text Screening}
\author*[1]{\fnm{Marco} \sur{Mandap}}
\email{marco.mandap@bulsu.edu.ph}
\author[1]{\fnm{Jerahmeel} \sur{Hipolito}}
\email{jerahmeel.hipolito@bulsu.edu.ph}
\author[1]{\fnm{Arcel} \sur{Galvez}}
\email{arcel.galvez@bulsu.edu.ph}
\author[1]{\fnm{Charlie Margaret} \sur{Balagtas}}
\email{2023102682@ms.bulsu.edu.ph}
\author[1]{\fnm{Michael Joshua} \sur{Buluran}}
\email{2023102675@ms.bulsu.edu.ph}
\author[1]{\fnm{Jeff Roel} \sur{Durmiendo}}
\email{2023102652@ms.bulsu.edu.ph}
\author[1]{\fnm{Rizzette E.} \sur{Lopez}}
\email{2022104397@ms.bulsu.edu.ph}
\affil[1]{\orgname{Bulacan State University}}

\abstract{We study false-alert control when screening for text generated by artificial intelligence (AI). The screening procedure selects document prefixes and detectors from observed evidence and may stop before exhausting its inspection budget. We give two finite-sample constructions under document-level exchangeability between human calibration documents and a new null document, with no restriction on dependence among tokens within a document. Construction A registers a finite family of prefix--detector scores and allocates a false-alert budget across their conformal ranks. A union bound protects any executed subset of that family. Construction B calibrates the complete-path maximum of a development-fixed adaptive policy. Each partial-path maximum is bounded by the complete maximum, so a terminal conformal rank protects early stopping without splitting the error budget. We prove marginal control of any false alert across the permitted inspection path and derive necessary calibration counts for rejection. We also state oracle testing, distribution-shift, and independent-audit bounds with their additional assumptions. Both constructions protect stopping within their specified scope; neither proof constructs an e-process or justifies multiplying conformal ranks. Detection power and computational savings remain questions for empirical evaluation.}

\keywords{AI-generated text detection, sequential testing, conformal calibration, exchangeability, multiple testing, false-alert control}
\maketitle

\section{Introduction}

Researchers screening for text generated by artificial intelligence (AI) must choose how much to inspect and which detector to run next. They can use probability-curvature scores from DetectGPT and Fast-DetectGPT or contrasting-model scores from Binoculars \citep{mitchell2023detectgpt,bao2024fastdetectgpt,hans2024binoculars}. To preserve a false-alert target while inspecting further prefixes or switching detectors, they need an error bound for the whole decision path. A guarantee for one fixed score does not justify selecting the most suspicious result from several inspections.

We treat a rejection as a screening alert against a specified human-reference distribution. Human authors can write text outside that population. An alert alone therefore cannot establish AI authorship or misconduct. We retain an insufficient-evidence status for a document that receives no alert.

\subsection{Positioning}

Chakraborty et al.\ analyze AI-text detectability under several settings, including an independent-replicate formulation \citep{chakraborty2024position}. Their sample counts do not specify a universal within-document token budget. Sadasivan et al.\ study distributional similarity and attacks that weaken detection \citep{sadasivan2025canai}. We distinguish these limits on power from the error-control question studied here.

Zhu et al.\ use length-aware multiscaled conformal prediction (MCP) to calibrate detection scores \citep{zhu2025reliably}. A threshold for one submitted text length does not, by itself, bound the probability of crossing any threshold across several inspected prefixes and detectors. Chen and Wang study sequential source detection from a stream of texts \citep{chenwang2025online}. Their stream-level setting differs from inspecting dependent segments of one document. We adopt the distinction between marginal and calibration-conditional validity in conformal outlier testing \citep{bates2023testing}.

Sequential testing predates AI-text detection \citep{wald1945sequential}. Howard et al.\ give time-uniform bounds through nonnegative supermartingales \citep{howard2020timeuniform}, and Koning and van Meer study how to induce sequential tests from terminal tests \citep{koning2026anytime}. Thus, deriving stopping-valid decisions from fixed-decision guarantees is not a new general principle. Our contribution is a specification of two such constructions for within-document detector selection, including the restrictions on routing and calibration reuse.

We restrict the application to passive detection of unwatermarked text. Generation-side watermarks \citep{kirchenbauer2023watermark} and conformal assessment of watermark-based AI assistance \citep{xie2026watermark} address related decisions with different evidence. We make no claim to originate conformal AI-use detection or foundational sequential-testing theory.

AdaDetectGPT learns a witness function from training data \citep{zhou2025adadetect}. Its adaptation concerns score learning, whereas our policies choose inspections within a test document after development. A learned score could enter our registered family, but its fixed-decision guarantee would not replace the pathwise error accounting below.

\subsection{Contributions}

We give two finite-sample rules with a shared document-level false-alert target. Construction A registers prefix--detector actions and fixes their error allocations before calibration. Its union-bound guarantee permits score-dependent action selection within that family (Theorem~\ref{thm:constructionA}). Construction B freezes an adaptive route generator and calibrates its complete-path maximum. Its guarantee permits early termination along that route without splitting the error budget (Theorem~\ref{thm:constructionB}).

For each construction, we derive the calibration count needed to permit rejection on the conformal rank grid. We distinguish this requirement from useful detection power and from offline scoring cost. Supporting results state the extra assumptions needed for oracle sample budgets, distribution-shift bounds, and independent audits. We also separate the stopping control proved here from an e-process construction. No detector experiments or computational savings are reported.

\section{Framework}\label{sec:framework}

Let \(X=(W_1,\ldots,W_L)\) denote one complete document, including its length and acquisition metadata. Fix a human-document reference law \(P_0\) and a family \(\mathcal Q\) of generated or AI-assisted alternatives. Under the null, \(X\sim P_0\). Rejecting \(P_0\) does not identify a member of \(\mathcal Q\), because other human populations can also differ from \(P_0\).

A policy \(\pi\) chooses an inspection extent and detector at each step and terminates within a finite action budget. It reports \emph{flag for review} if it crosses the alert boundary. Otherwise, including after a futility stop, it reports \emph{no alert at this budget}. We do not certify human authorship after non-rejection.

Write \(\tau\) for the number of executed actions, \(N_\tau\) for the number of distinct tokens inspected, and \(C_\tau\) for total computational cost. Several detectors may rescore the same tokens, so these quantities need not agree. For \(0<\alpha<1\), our target is
\begin{equation}
 \Prob_0\{\text{an alert occurs at any allowed action}\}\leq\alpha.
 \label{eq:target}
\end{equation}
Here \(\Prob_0\) denotes the joint null experiment, including development information, calibration documents, the new document, and algorithmic randomness. This is marginal document-level control. It does not condition on a realized calibration set, document length, subgroup, or selected route, and it does not control multiplicity across a collection of screened documents.

\begin{assumption}[Development separation and exchangeability]\label{ass:exch}
Fix score maps, model checkpoints, preprocessing, action families, transformations, and error weights using development information \(\Ddev\). Conditional on \(\Ddev\), the human calibration documents \(H_1,\ldots,H_m\) and a new null document \(X\) are exchangeable. Apply per-document randomization symmetrically, for example through independent seeds drawn from a common distribution independently of the documents. For Construction B, also fix the complete-route generator using \(\Ddev\), without using the held-out calibration documents or their summaries.
\end{assumption}

We take \(m\geq1\) and assume all score and selection maps are measurable. Assumption~\ref{ass:exch} permits arbitrary token dependence within a document. For randomized scores, we augment each document with its seed and apply the same fixed map to the augmented documents. These augmented documents remain exchangeable. Shared randomness fixed during development belongs to \(\Ddev\).

The proofs condition on \(\Ddev\) and average over it at the end. They do not invoke conformal validity conditional on a realized calibration sample. Construction A allows action selection from observed scores or calibration ranks; its fixed family and weights, rather than a calibration-independent route, provide the guarantee. Construction B requires the stronger routing restriction in Assumption~\ref{ass:exch}.

\section{A conformal ranking lemma}\label{sec:prelim}

\begin{lemma}[Conservative upper-tail ranks]\label{lem:ranks}
Let \(Z_1,\ldots,Z_{m+1}\) be exchangeable conditional on \(\Ddev\), with scores in the extended real line. Define
\begin{equation}
 R=1+\sum_{i=1}^m\ind\{Z_i\geq Z_{m+1}\},
 \qquad p=\frac{R}{m+1}.
 \label{eq:rank}
\end{equation}
For \(u\in[0,1]\),
\begin{equation}
 \Prob(p\leq u\mid\Ddev)
 \leq\frac{\lfloor(m+1)u\rfloor}{m+1}\leq u.
 \label{eq:rankbound}
\end{equation}
\end{lemma}

\begin{proof}
Introduce independent \(V_1,\ldots,V_{m+1}\sim\operatorname{Unif}(0,1)\), independent of the scores and \(\Ddev\). Order \((Z_i,V_i)\) lexicographically from largest to smallest. The pairs remain exchangeable and have distinct ranks with probability one. The descending rank \(R^*\) of \((Z_{m+1},V_{m+1})\) is therefore uniform on \(\{1,\ldots,m+1\}\), conditional on \(\Ddev\).

The conservative rank \(R\) counts all calibration scores tied with the test score. The tie-broken rank counts a subset of those ties, so \(R\geq R^*\) almost surely. Hence
\begin{equation*}
 \begin{aligned}
 \Prob(p\leq u\mid\Ddev)
 &\leq\Prob\{R^*\leq\lfloor(m+1)u\rfloor\mid\Ddev\}\\
 &=\frac{\lfloor(m+1)u\rfloor}{m+1}\leq u.
 \end{aligned}
\end{equation*}
\end{proof}

This is the conservative conformal rank argument used in outlier calibration \citep{bates2023testing}. The auxiliary uniforms serve the proof; the implemented rank is deterministic given the scores. If all scores tie, \(p=1\). A constant score therefore cannot trigger an alert at \(\alpha<1\), regardless of the calibration count.

\section{Construction A: a registered family with an error budget}\label{sec:constructionA}

Using \(\Ddev\), fix a finite family \(\mathcal A\) of actions \(a=(k,j)\), where \(k\) indexes a token budget \(b_k\) and \(j\) a detector. Compute \(S_a(X)\) from the first \(\min(b_k,L)\) tokens, under a fixed inspection tokenizer and a fixed map with larger scores indicating greater suspicion. Use the same truncation and failure rules for calibration and test documents. For example, \(S_a(X)=-\infty\) on failure gives \(p_a(X)=1\) and prevents an alert at that action. Define each score on all documents, including documents on which a deployed run does not execute the action.

For \(a\in\mathcal A\), form
\begin{equation}
 p_a(X)=\frac{1+\sum_{i=1}^m\ind\{S_a(H_i)\geq S_a(X)\}}{m+1}.
 \label{eq:pa}
\end{equation}
Fix \(w_a\geq0\) with \(\sum_{a\in\mathcal A}w_a\leq1\), and set \(\alpha_a=\alpha w_a\). During deployment, choose an executed subset \(\mathcal E\subseteq\mathcal A\) and its order from observed scores or ranks. Alert when an executed action satisfies \(p_a\leq\alpha_a\). No other event authorizes an alert under this rule.

Keep the allocations fixed after observing calibration or test scores. Reuse an action's realized score or frozen seed if the action is revisited. A fresh randomized replicate requires its own registered action and allocation.

\begin{theorem}[False-alert control under action selection]\label{thm:constructionA}
Under Assumption~\ref{ass:exch}, Construction A satisfies \eqref{eq:target} for any measurable rule selecting the executed subset and its order from the registered family. The action scores may be dependent.
\end{theorem}

\begin{proof}
Conditional on \(\Ddev\), fix \(a\in\mathcal A\). Applying Lemma~\ref{lem:ranks} to its calibration and test scores gives
\begin{equation}
 \Prob_0(p_a\leq\alpha_a\mid\Ddev)\leq\alpha_a=\alpha w_a.
 \label{eq:perAction}
\end{equation}
On each realization,
\begin{equation}
 \begin{aligned}
 \{\text{alert}\}
 &\subseteq\bigcup_{a\in\mathcal A}
       \bigl(\{a\in\mathcal E\}\cap\{p_a\leq\alpha_a\}\bigr)\\
 &\subseteq\bigcup_{a\in\mathcal A}\{p_a\leq\alpha_a\}.
 \end{aligned}
 \label{eq:containment}
\end{equation}
The final union ranges over the fixed family, regardless of execution. Taking conditional probabilities and applying \eqref{eq:perAction},
\begin{equation*}
 \begin{aligned}
 \Prob_0(\text{alert}\mid\Ddev)
 &\leq\sum_{a\in\mathcal A}\Prob_0(p_a\leq\alpha_a\mid\Ddev)\\
 &\leq\alpha\sum_{a\in\mathcal A}w_a\leq\alpha.
 \end{aligned}
\end{equation*}
Average over \(\Ddev\) to obtain \eqref{eq:target}. We use neither independence between actions nor a rank bound conditional on execution.
\end{proof}

\begin{remark}[Selection versus reallocation]
Theorem~\ref{thm:constructionA} applies weighted Bonferroni testing to fixed conformal ranks. The executed-subset containment preserves that guarantee under action selection; it is not a new multiple-testing inequality. The theorem does not justify choosing the weights after observing ranks. For independent null \(p\)-values \(U_1,U_2\sim\operatorname{Unif}(0,1)\), assign all weight to the smaller value and reject at level \(\alpha\). The error is \(1-(1-\alpha)^2>\alpha\). Although the selected weights sum to one, the thresholds depend on the evidence.
\end{remark}

\begin{corollary}[Calibration resolution]\label{cor:resolutionA}
The rank in \eqref{eq:pa} is at least \(1/(m+1)\). For \(w_a>0\), rejection through action \(a\) requires
\begin{equation}
 m\geq\left\lceil\frac{1}{\alpha w_a}\right\rceil-1.
 \label{eq:resA}
\end{equation}
An action with \(w_a=0\) cannot reject for any finite \(m\).
\end{corollary}

\begin{proof}
The numerator of \eqref{eq:pa} is at least one. Thus, rejection requires \(1/(m+1)\leq\alpha w_a\). Solving for the integer \(m\) gives \eqref{eq:resA}. For \(w_a=0\), the threshold is zero and the rank is positive.
\end{proof}

For twelve equally weighted actions, such as four token budgets and three detectors, rejection requires \(m\geq1{,}199\) at \(\alpha=0.01\), or \(m\geq11{,}999\) at \(\alpha=0.001\). These counts permit rejection on the rank grid; they do not ensure that a score attains the minimum rank or has useful power. Treating dependent prefixes from one document as separate calibration units would require a different exchangeability argument.

\section{Construction B: calibrating the complete policy path}\label{sec:constructionB}

Construction B uses one calibrated scalar instead of allocating error across actions. Using \(\Ddev\), fix a measurable finite-horizon route generator \(\pi\) and transformations \(g_a\). For successful scores, \(g_a\) may be the negative logarithm of a positive tail rank computed against a separate development reference set. Reserve \(-\infty\) as the failure sentinel and define \(g_a(-\infty)=-\infty\). Check for failure before evaluating the successful-score transformation. Apply this convention to calibration and test documents: a failed action must not add alert evidence.

For a document \(X\), define the complete route by running \(\pi\) to its planned terminal action \(T_\pi(X)\), with alert-based stopping disabled and development-fixed futility rules retained. Let
\begin{equation}
 M_\pi(X)=\max_{1\leq t\leq T_\pi(X)}
       g_{a_t(X)}\{S_{a_t(X)}(X)\}.
 \label{eq:Mpi}
\end{equation}
Set \(M_\pi(X)=-\infty\) for an empty route. The next action may depend on document features and scores observed so far. The route generator, transformations, and randomization scheme must not use the held-out calibration documents or their summaries.

Compute the complete route for each calibration document. At deployment, follow an initial segment of the prescribed route. Early termination may depend on a calibrated alert or a resource or futility decision, but calibration information must not change the next action, transformation, or seed. The complete test route is a mathematical reference for the proof; deployment need not compute its unvisited actions.

For an executed action \(t\), write
\begin{equation*}
 M_{\pi,t}(X)=\max_{1\leq s\leq t}
       g_{a_s(X)}\{S_{a_s(X)}(X)\},
\end{equation*}
and compare this partial maximum with the complete calibration maxima:
\begin{equation}
 p_t^{\mathrm{path}}(X)=
 \frac{1+\sum_{i=1}^m\ind\{M_\pi(H_i)\geq M_{\pi,t}(X)\}}{m+1}.
 \label{eq:ppath}
\end{equation}
Alert at the first executed action with \(p_t^{\mathrm{path}}\leq\alpha\). An empty deployed route produces no alert.

\begin{theorem}[Complete-path conformal stopping]\label{thm:constructionB}
Under Assumption~\ref{ass:exch}, let \(M_\pi\) be the fixed measurable complete-path score in \eqref{eq:Mpi}, using the same per-document randomization scheme for calibration and test documents. If deployment follows a prefix of that finite route, Construction B satisfies \eqref{eq:target} for arbitrary early termination along the route.
\end{theorem}

\begin{proof}
Applying the same map \(M_\pi\) to the augmented calibration and test documents preserves exchangeability conditional on \(\Ddev\). By Lemma~\ref{lem:ranks}, the terminal rank
\begin{equation*}
 p^{\mathrm{full}}(X)=
 \frac{1+\sum_{i=1}^m\ind\{M_\pi(H_i)\geq M_\pi(X)\}}{m+1}
\end{equation*}
satisfies \(\Prob_0(p^{\mathrm{full}}\leq\alpha\mid\Ddev)\leq\alpha\).

Path consistency gives \(M_{\pi,t}(X)\leq M_\pi(X)\) at each executed action. A lower test score has at least as many dominating calibration scores, so
\begin{equation*}
 p_t^{\mathrm{path}}(X)\geq p^{\mathrm{full}}(X).
\end{equation*}
Writing \(\rho=\tau\leq T_\pi(X)\) for the deployed stopping index, where \(\tau\) is the executed-action count from Section~\ref{sec:framework},
\begin{equation}
 \{\exists t\leq\rho:p_t^{\mathrm{path}}(X)\leq\alpha\}
 \subseteq\{p^{\mathrm{full}}(X)\leq\alpha\}.
 \label{eq:containB}
\end{equation}
The left event is empty when \(\rho=0\). Otherwise the containment bounds its conditional probability by \(\alpha\). Averaging over \(\Ddev\) proves \eqref{eq:target}. No optional-stopping theorem for the partial ranks is needed.
\end{proof}

\begin{remark}[Calibration reuse]
Calibrating against matching prefix maxima instead of complete maxima changes the reference scores and loses the containment argument in \eqref{eq:containB}. A changed policy needs calibration scores from its own complete routes. Recomputing on the same calibration documents is valid under this argument if the new policy was chosen without those documents or their summaries. Calibration-informed policy tuning requires fresh independent calibration or another validity argument.
\end{remark}

\paragraph*{Empty routes and failed actions.}
An empty route produces no alert. If all executed actions fail, then \(M_{\pi,t}(X)=-\infty\), so every complete calibration maximum dominates it and \(p_t^{\mathrm{path}}(X)=1\). In particular, with \(m=99\) empty calibration routes and one failed test action, the rank is one, not \(1/100\). Evaluating a failure through a tail-rank transformation without the sentinel override could map it to zero and produce the latter rank. A failure after a successful action leaves the running maximum unchanged; the development-fixed policy may continue or stop without an alert.

\paragraph*{Relationship to induced sequential tests.}
Koning and van Meer's Theorem~1 induces a sequential test by conditioning a terminal test on the observed information under a specified null law \citep{koning2026anytime}. To compare, pad the complete route to its fixed finite horizon and let \(\mathcal F_t\) contain development information, calibration data, and the observations available through action \(t\). Set
\[
 \phi_t=\ind\{p_t^{\mathrm{path}}\leq\alpha\},
 \qquad \phi_{\mathrm{full}}=\ind\{p^{\mathrm{full}}\leq\alpha\}.
\]
Use \(\phi_t=0\) before any action and retain the terminal value after the complete route. For each joint null law \(\mathsf P\) satisfying our assumptions, containment and \(\mathcal F_t\)-measurability give
\[
 \phi_t\leq\Expect_{\mathsf P}(\phi_{\mathrm{full}}\mid\mathcal F_t).
\]
Thus B implements an observable conservative lower bound on that induced test, not its conditional expectation. Our proof does not require evaluating a joint null model or conditioning the false-alert target on the realized calibration set. The contribution is this complete-path conformal implementation with its routing restrictions, not a new induction principle. We claim neither equality with the induced test nor admissibility or unrestricted optional continuation.

\begin{corollary}[Calibration resolution]\label{cor:resolutionB}
Rejection in Construction B requires \(m\geq\lceil1/\alpha\rceil-1\).
\end{corollary}

\begin{proof}
Equation~\eqref{eq:ppath} is at least \(1/(m+1)\). Rejection therefore requires \(1/(m+1)\leq\alpha\), giving the stated integer bound.
\end{proof}

At \(\alpha=0.01\) and \(0.001\), the corresponding counts are \(99\) and \(999\). Compared with twelve equal allocations in Construction A, these are lower rank-resolution requirements. They do not establish lower offline cost or greater power. Construction B requires complete calibration routes, and a broader path can raise its calibrated maximum. A constant complete-path score still gives \(p_t^{\mathrm{path}}=1\) at each executed action.

\section{Supporting bounds}\label{sec:supporting}

\subsection{An oracle total-variation identity}

\begin{theorem}[Oracle testing bound]\label{thm:oracle}
For two document laws \(P_0,Q\), let \(P_n=P_0^{\otimes n}\) and \(Q_n=Q^{\otimes n}\). Then
\begin{equation}
 \inf_{0\leq\phi\leq1}
 \{\Expect_{P_n}\phi+\Expect_{Q_n}(1-\phi)\}
 =1-\TV(P_n,Q_n),
 \label{eq:oracle}
\end{equation}
where the infimum ranges over measurable randomized rejection rules and
\(\TV(P,Q)=\sup_B|P(B)-Q(B)|\).
\end{theorem}

\begin{proof}
For a deterministic rule \(\phi=\ind_B\), total error equals
\begin{equation*}
 P_n(B)+Q_n(B^c)=1-\{Q_n(B)-P_n(B)\}.
\end{equation*}
Let \(\lambda=P_n+Q_n\), \(f=\mathrm{d}P_n/\mathrm{d}\lambda\), and \(g=\mathrm{d}Q_n/\mathrm{d}\lambda\). Since \(\int(g-f)\,\mathrm{d}\lambda=0\), the positive and negative parts have equal integrals. Choosing \(B^*=\{g>f\}\) therefore gives
\begin{equation*}
 \sup_B\{Q_n(B)-P_n(B)\}
 =\int(g-f)_+\,\mathrm{d}\lambda=\TV(P_n,Q_n).
\end{equation*}
For a randomized rule \(0\leq\phi\leq1\),
\begin{equation*}
 \begin{aligned}
 \Expect_{P_n}\phi+\Expect_{Q_n}(1-\phi)
 &=1-\int\phi(g-f)\,\mathrm{d}\lambda\\
 &\geq1-\TV(P_n,Q_n).
 \end{aligned}
\end{equation*}
Equality holds at \(\phi=\ind_{B^*}\).
\end{proof}

The identity holds for any two joint laws. Independence enters when we write those laws as products. It does not convert the number of tokens in one document into an independent-replicate count. A learned detector may also discard information used by the oracle rule.

\subsection{A Hoeffding budget under independence}

\begin{proposition}[Illustrative oracle budget]\label{prop:hoeffding}
Let \(Z_1,\ldots,Z_n\in[0,1]\) be independent under each hypothesis, with common null mean \(\mu_0\) and common alternative mean \(\mu_0+\Delta\), where \(\Delta>0\). Suppose these means are known. Reject if \(\bar Z_n\geq\mu_0+\Delta/2\). Each error probability is at most \(\exp(-n\Delta^2/2)\). For \(0<\alpha,\beta<1\),
\begin{equation}
 n\geq\frac{2}{\Delta^2}
       \max\!\left\{\log\frac1\alpha,\log\frac1\beta\right\}
 \label{eq:hoeffbudget}
\end{equation}
suffices for level \(\alpha\) and power at least \(1-\beta\).
\end{proposition}

\begin{proof}
Hoeffding's inequality \citep{hoeffding1963probability} gives
\begin{equation*}
 \Prob(\bar Z_n-\Expect\bar Z_n\geq\varepsilon)
 \leq \exp(-2n\varepsilon^2),
\end{equation*}
with the same bound for the lower tail. Under the null, rejection requires an upward deviation of at least \(\Delta/2\), so its probability is at most
\begin{equation*}
 \exp\{-2n(\Delta/2)^2\}=\exp(-n\Delta^2/2).
\end{equation*}
Under the alternative, non-rejection requires a downward deviation greater than \(\Delta/2\); the lower-tail bound applies. Requiring the two bounds to be at most \(\alpha\) and \(\beta\), respectively, gives \eqref{eq:hoeffbudget}, rounded up to an integer.
\end{proof}

Independent bounded scores with the stated means suffice; identical score distributions are unnecessary. Applying this budget to learned conformal rules would require accounting for mean estimation, score selection, and repeated inspection. Applying it to tokens would require assumptions on their dependence beyond Assumption~\ref{ass:exch}.

\subsection{Distribution shift}

\begin{proposition}[Shift bound averaged over calibration]\label{prop:shift}
Let \(C\) collect the fitted calibration mechanism and algorithmic randomness, with law \(\mu\). Suppose the joint reference and shifted experiments are \(\mu\otimes P_0\) and \(\mu\otimes\widetilde P_0\), so the test document is independent of \(C\) under both laws. Assume reference validity:
\begin{equation*}
 \Prob_{\mu\otimes P_0}(\text{alert})\leq\alpha.
\end{equation*}
For a measurable alert event with document section \(A_c\) at \(C=c\),
\begin{equation}
 \Prob_{\mu\otimes\widetilde P_0}(\text{alert})
 \leq\min\{1,\alpha+\TV(P_0,\widetilde P_0)\}.
 \label{eq:shift}
\end{equation}
\end{proposition}

\begin{proof}
Set \(q_0(c)=P_0(A_c)\), \(\widetilde q(c)=\widetilde P_0(A_c)\), and \(\delta=\TV(P_0,\widetilde P_0)\). For each \(c\), the definition of total variation gives \(\widetilde q(c)\leq q_0(c)+\delta\). Integrating,
\begin{equation*}
 \begin{aligned}
 \Prob_{\mu\otimes\widetilde P_0}(\text{alert})
 &=\int\widetilde q(c)\,\mu(\mathrm{d}c)\\
 &\leq\int q_0(c)\,\mu(\mathrm{d}c)+\delta\leq\alpha+\delta.
 \end{aligned}
\end{equation*}
The left side is also at most one.
\end{proof}

Theorems~\ref{thm:constructionA} and~\ref{thm:constructionB} supply reference validity when their assumptions hold for the reference experiment. They supply an average bound, not \(q_0(c)\leq\alpha\) for each realized calibration mechanism. The product-law condition in Proposition~\ref{prop:shift} is an additional assumption: document-level exchangeability alone does not imply independence of a new document and its calibration set. More general departures from exchangeability need further structure \citep{barber2023conformal}.

\begin{example}[Marginal versus calibration-conditional error]\label{ex:conditional}
Let \(H,X\) be independent \(\operatorname{Unif}(0,1)\) variables and use the identity score \(S(z)=z\). With one calibration score and \(\alpha=1/2\),
\[
 p(X)=\frac{1+\ind\{H\geq X\}}{2},
 \qquad \{p(X)\leq1/2\}=\{X>H\}.
\]
Its marginal error is \(\Prob(X>H)=1/2\). Conditional on \(H=h\), its error is \(1-h\), which exceeds \(1/2\) for \(h<1/2\). For a fixed mechanism \(c\), the shift argument gives \(\min\{1,q_0(c)+\delta\}\), not \(\min\{1,\alpha+\delta\}\).
\end{example}

The shift bound holds as an inequality even if \(\delta\) is unknown. A numerical deployment certificate requires a justified upper bound on that distance, obtained from assumptions or data. A small deployment sample alone does not provide such a certificate for an unrestricted document distribution.

\subsection{A zero-event audit bound}

\begin{proposition}[Clopper--Pearson audit bound]\label{prop:audit}
Fix a fitted pipeline and an audit size \(N_0\geq1\) before observing audit outcomes. Use \(N_0\) independent and identically distributed human audit documents, independent of fitting data. If scoring is randomized, use independent per-document randomization from a common fixed law. Conditional on the fitted pipeline, the false-alert count is \(K\sim\operatorname{Binomial}(N_0,q)\), where \(q\) is its deployment false-alert probability. For \(K=0\), the one-sided Clopper--Pearson upper confidence limit at confidence level \(1-\gamma\), \(0<\gamma<1\), is
\begin{equation}
 U_{1-\gamma}=1-\gamma^{1/N_0}.
 \label{eq:audit}
\end{equation}
At \(\gamma=0.05\), the condition \(U_{0.95}<0.001\) requires \(N_0\geq2{,}995\).
\end{proposition}

\begin{proof}
In the zero-event case, binomial tail inversion \citep{clopperpearson1934} solves
\begin{equation*}
 \Prob_{q=U}(K=0)=(1-U)^{N_0}=\gamma,
\end{equation*}
which gives \eqref{eq:audit}. At \(\gamma=0.05\),
\begin{equation*}
 U_{0.95}<0.001
 \quad\Longleftrightarrow\quad
 N_0>\frac{\log(0.05)}{\log(0.999)}\approx2994.234.
\end{equation*}
The smallest qualifying integer is \(2{,}995\).
\end{proof}

The full Clopper--Pearson procedure has coverage at least \(1-\gamma\) over repeated independent audits of fixed size; \eqref{eq:audit} describes its output when \(K=0\). It gives neither a posterior probability for \(q\) nor a prediction of zero observed errors. Outcome-dependent audit stopping needs a sequential confidence procedure. Shared authorship or prompt lineage can also invalidate the independent-binomial model.

\subsection{Repeated looks, stopping control, and e-processes}\label{sec:eprocess}

For \(K_{\mathrm{look}}\) independent uniform \(p\)-values \(U_1,\ldots,U_{K_{\mathrm{look}}}\), rejecting at any crossing of \(\alpha\) gives
\begin{equation}
 \Prob\!\left(\min_{1\leq k\leq K_{\mathrm{look}}}U_k\leq\alpha\right)
 =1-(1-\alpha)^{K_{\mathrm{look}}}.
 \label{eq:inflation}
\end{equation}
Under dependence, this value can change. Identical uniform \(p\)-values give error \(\alpha\). Without independence, the union bound still gives at most \(\min\{1,K_{\mathrm{look}}\alpha\}\). We retain \(K\) for the random audit count in Proposition~\ref{prop:audit}.

Theorems~\ref{thm:constructionA} and~\ref{thm:constructionB} do establish simultaneous control across their allowed inspections. They therefore protect data-dependent stopping within the registered family or frozen route. Calling these guarantees merely fixed-time, or claiming that they provide no time-uniform control, would understate the results. Their scope remains one document and the specified action family or finite route.

A nonnegative test supermartingale provides another route to stopping-valid inference \citep{howard2020timeuniform}. Relative to the available-information filtration \((\mathcal F_t)\), suppose \((E_t)\) is adapted and integrable, with
\begin{equation}
 E_0=1,\qquad E_t\geq0,\qquad
 \Expect_0(E_t\mid\mathcal F_{t-1})\leq E_{t-1}.
 \label{eq:supermartingale}
\end{equation}
For \(c>0\), let \(\sigma_c=\inf\{t\geq0:E_t\geq c\}\). At a finite horizon \(T\), bounded optional stopping and nonnegativity give
\begin{equation*}
 1\geq\Expect_0 E_{\sigma_c\wedge T}
 \geq c\,\Prob_0(\sigma_c\leq T).
\end{equation*}
Letting \(T\to\infty\) yields \(\Prob_0(\sigma_c<\infty)\leq1/c\). For \(c<1/\alpha\), the event \(\{\sup_t E_t\geq1/\alpha\}\) implies a crossing of \(c\). Let \(c\uparrow1/\alpha\) to obtain
\begin{equation}
 \Prob_0\!\left(\sup_{t\geq0}E_t\geq1/\alpha\right)\leq\alpha,
 \label{eq:ville}
\end{equation}
including a supremum that reaches the boundary only as a limit.

Condition~\eqref{eq:supermartingale} is sufficient, not a definition of all e-processes. More generally, an adapted nonnegative process is an e-process for a null family if its expectation at each permitted stopping time is at most one under each null law. For this definition, see Ramdas, Ruf, Larsson, and Koolen (2022, preprint, \href{https://arxiv.org/abs/2009.03167v3}{arXiv:2009.03167v3}). Our rank proofs construct no such process and establish no conditional supermartingale property for products of transformed ranks.

Marginal rank validity alone cannot justify those products. For example, \(U\sim\operatorname{Unif}(0,1)\) is a valid \(p\)-value, but \(\Expect(1/U)=\infty\). Even for transformations with bounded marginal expectation, dependence between successive factors requires a separate sequential argument. Thus the result established here is pathwise false-alert control, not a betting-product rule or permission for unrestricted optional continuation.

\section{Discussion: scope and limits}\label{sec:discussion}

\paragraph*{Distribution shift.}
Proposition~\ref{prop:shift} keeps the calibration mechanism's law unchanged and changes an independent test draw. It averages over calibration. Example~\ref{ex:conditional} rules out replacing that average guarantee with a level-\(\alpha\) guarantee for each realized calibration set. A numerical shifted-error bound also needs control of \(\TV(P_0,\widetilde P_0)\).

\paragraph*{Subgroup validity.}
A pooled marginal bound does not ensure protection within a subgroup. A stratum-specific conformal argument requires exchangeability within that stratum and sufficient calibration documents there. Researchers must also have authorized access to the attributes used for stratification. Detector bias against non-native English writing illustrates why pooled performance alone is inadequate \citep{liang2023gptdetectors}.

\paragraph*{Adversarial selection and mixed authorship.}
Selecting a human document after querying a detector changes the test-selection mechanism and can invalidate the null exchangeability assumption. Paraphrasing generated text raises a separate power issue: an adversary can weaken evidence against an alternative \citep{sadasivan2025canai}. Mixed human/AI authorship requires a specified target population and label definition; the binary reference test does not estimate the extent of AI assistance.

\paragraph*{Deployment-wide multiplicity.}
Equation~\eqref{eq:target} concerns one null-document decision. Screening many documents, or rerunning one document under several policies, requires additional multiplicity accounting if the target is deployment-wide family-wise error or false-discovery control. The within-document theorems do not supply that accounting.

\paragraph*{Score construction.}
Conformal calibration can use a prespecified measurable scalar score without assuming independent tokens. It does not validate an uncalibrated word-level test. For example, the standard two-sample Kolmogorov--Smirnov calibration assumes independent samples from continuous distributions \citep{scipy2024ks2samp}. Arbitrary vocabulary identifiers impose a relabeling-dependent order, while repeated words create ties. A document-level discrepancy needs its own scientific interpretation and document-level calibration.

\paragraph*{Interpretation and usefulness.}
An alert concerns departure from the human reference, not a posterior probability of AI authorship. Non-rejection retains an insufficient-evidence status. The proofs also allow valid but powerless scores. An empirical comparison must measure detection power and computational cost at matched error targets, including calibration work and abstentions.

\section{Conclusion}\label{sec:conclusion}

We give two conformal rules for adaptive inspection within one document. Construction A fixes an action family and its error allocations, then permits selection from that family. Construction B fixes a route generator and calibrates its complete-path maximum, then permits early termination along the route. Both control the marginal probability of any false alert under document-level exchangeability and their stated development restrictions, without assuming independent tokens.

The calibration counts describe rank feasibility, not sufficient data for useful power. The oracle, shift, and audit results require additional assumptions stated with each result. Neither stopping proof constructs an e-process, although both provide stopping-valid decisions within their prescribed scope. A companion empirical study must determine whether either rule offers useful power and computation savings on text.

\backmatter
\section*{Statements and Declarations}

\textbf{Manuscript status.} This theory manuscript contains analytic results and no empirical evaluation. We present the conformal ranking, union-bound, oracle, concentration, and audit arguments as established tools applied to the specified screening constructions. We do not claim new foundational theory or an exhaustive novelty review. The supporting propositions use the additional assumptions stated in each result.

\textbf{Funding.} No funding was received for this work.

\textbf{Competing interests.} The authors must confirm their relevant financial and non-financial interests before submission.

\textbf{Author contributions.} The authors confirm the following CRediT roles: M.\,Mandap, conceptualization, supervision, and review editing; J.\,Hipolito, methodology and review editing; A.\,Galvez, supervision and review editing; C.\,M.\,Balagtas, conceptualization and original draft; M.\,J.\,Buluran, methodology and software; J.\,R.\,Durmiendo, validation and investigation; R.\,E.\,Lopez, supervision and review editing. Marco Mandap served as the primary author. All authors reviewed and approved the final manuscript.

\textbf{Ethics approval and consent.} This article reports theoretical analysis only, with no studies involving human participants or animals and no participant data. The authors must confirm any institutional requirements before submission.

\textbf{Data availability.} No datasets were generated or analysed in this theoretical study.

\textbf{Code availability.} The Project Burnwatch prototype is available in the versioned repository at \url{https://github.com/mrphilo420/Burnwatch/tree/v0.1.0}. The mathematical results are stated and proved in the text.

\textbf{AI assistance.} AI tools assisted with literature synthesis, drafting, editing, and mathematical review. The authors are responsible for the manuscript and must verify the claims, proofs, references, and final text before submission.

\end{document}